\documentclass[conference]{IEEEtran}
\usepackage{amssymb}
\usepackage{amsmath}
\usepackage{amsfonts}
\usepackage{graphicx}
\usepackage{algorithm}
\usepackage{algorithmic}
\usepackage{cite}
\usepackage{epstopdf}

\newtheorem{theorem}{Theorem}

\newtheorem{proposition}[theorem]{\textbf{Proposition}}

\IEEEoverridecommandlockouts
\begin{document}
\title{Opportunistic Channel Sharing in Stochastic Networks with Dynamic Traffic}
\author{Bin Xia$^{\dag}$, Yingbin Liu$^{\dag}$, Chenchen Yang$^{\dag}$, Zhiyong Chen$^{\dag}$, Weiliang Xie$^{\ddag}$ and Yong Zhao$^{\ddag}$\\
$^\dag$Department of Electronic Engineering, Shanghai Jiao Tong University, Shanghai, P. R. China\\
$^{\ddag}$ Technology Innovation Center, China Telecom\\
Email: {\{bxia, tracylyb, zhanchifeixiang, zhiyongchen\}@sjtu.edu.cn,  \{xiewl, zhaoyong\}@ctbri.com.cn}
}
\maketitle
\begin{abstract}
In this paper, we consider the stochastic network with dynamic traffic. The spatial distribution of access points (APs) and users are first modeled as mutually independent Poisson point processes (PPPs).  Different from most previous literatures which assume all the APs are fully loaded, we consider the fact that APs having no data to transmit do not generate interference to users. The APs opportunistically share the channel according to the existence of the packet to be transmitted and the proposed interference suppression strategy. In the interference suppression region, only one AP can be active at a time  to transmit the packet on the channel and the other adjacent APs keep silent to reduce serious interference. The idle probability of any AP, influenced by the traffic load and availability of the channels, is analyzed.  The density of simultaneously active APs in the network is obtained  and the packet loss rate is further elaborated. We reveal the impacts of network features (e.g., AP density, user density and channel state) and service features (e.g., user request, packet size) on the network performance. Simulation results validate our proposed model.
\end{abstract}
\section{Introduction}

Due to the tremendous increase of wireless communication devices, cellular systems are facing an increasing demand for mobile data traffic. The dense deployment of  access points (APs) is one of the solutions \cite{paper14}. To analyze the performance of the cellular networks, Poisson Point Process (PPP) is a feasible tool of stochastic geometry  to model the positions of different nodes, e.g., APs and users. \cite{paper15} and \cite{paper16} characterize  the accuracy and tractability of PPP, while all cells are assumed to be fully loaded. Based on this assumption, related works have been conducted on stochastic networks. For example, \cite{paper18} proposes a PPP-based flexible and tractable model for multi-tier multiple-input-multiple-output heterogeneous networks, and the downlink coverage probability is evaluated. However, in previous works, under the ideal assumption of fully loaded cells, the interfering signals of a user come from all other APs in the network except for its serving AP.  In fact, APs having no data to  transmit do not generate interference to the user. 

Moreover, as APs transmit wireless signals simultaneously, interference degrades the system performance and should be reduced as much as possible. \cite{paper13} derives the signal-to-interference-plus-noise ratio (SINR) for the multi-cell with new general models. Statistical SINR is analyzed with the PPP approach. The authors in \cite{paper19} analyze higher order moments of SINR in a PPP-based heterogeneous interference field. However, in previous works,  one channel is assumed to be shared by all the fully loaded APs in the network all the time. The interference from the adjacent APs is serious, especially for the high-density network. 


In this paper, we consider opportunistic channel sharing. We propose an interference suppression region where only one AP can be active to transmit signals on the channel, to avoid serious interference from the adjacent APs. The interference model for the network with dynamic traffic is studied and the performance in terms of the packet loss rate (PLR) is derived. The main contributions are summarized as follows,
\begin{itemize}
\item \emph{An interference suppression strategy:} A region with a fixed radius for each AP is introduced to reduce serious interference from the adjacent APs. In the region of each AP, only one AP can transmit its packets on the channel and the other adjacent APs  keep silent without generating  interference. The probability that an AP fails to obtain the channel access opportunity to transmit the packet because of the interference suppression strategy is analyzed. Based on the probability, the PLR due to this factor is further derived.
\item \emph{Stochastic network with dynamic traffic:} The probability that an arbitrary AP is idle due to two factors (having no packets to transmit and failing to obtain the channel  access opportunity) is analyzed. Based on this, the simultaneously active APs are modeled as thinning  PPP to evaluate the interference under different traffic load. 
    In addition, due to the limited buffer size, packets will be dropped if  they are not transmitted in time. The PLR due to this factor is analyzed.
\end{itemize}

The remainder of this paper is organized as follows. The system model is introduced in Section  \uppercase\expandafter{\romannumeral2}. The AP-idle probability and PLR  are formulated in Section  \uppercase\expandafter{\romannumeral3}. Simulation results are provided in Section  \uppercase\expandafter{\romannumeral4}. Section  \uppercase\expandafter{\romannumeral5} concludes this paper.

\section{system model} \label{sys}
\subsection{Network Topology and Interference Suppression}
Consider that the APs and users are distributed in a large-scale 2-D plane according to two mutually independent homogeneous PPPs $\Phi_i = \{x_{i,j},j=1,2,3,\cdots\}$ with intensity $\lambda_i$, where $x_{i,j}$ denotes the position of the $j^{th}$ element of tier $i$ ($i=1,2$ for users and APs, respectively). Each user associates with the closest AP, namely the AP cells are polygonal and form the Voronoi tessellation on the plane.

One channel is shared in the network, and an interference suppression region is employed to avoid co-channel interference from adjacent APs. The interference suppression region is centered on each AP with a fixed radius $R$. In each interference suppression region, the active APs opportunistically share and compete for the channel with equal probability. Only one AP can transmit its packets at a time, and the other APs failing to access the channel in this region keep silent. Therefore, the simultaneously active APs have a minimum distance $R$, which can greatly reduce the interference. Accordingly, the area of the interference suppression region is $A=\pi R^2$.

\subsection{Transmission Protocol}
Consider a time-slotted system with slot duration $\tau$ $[$seconds$]$. We assume all the packets have the same size $\it{T}$ $[$bits$]$ and in each slot the AP is scheduled to transmit one packet. The packet request rate of each user is considered to be independent identically distributed. New request arrivals during a slot can not be delivered to the user until the beginning of the next slot. There are the following three cases when new packet requests arrive at the AP in a slot\footnote{Note: In this paper, we mainly focus on the state that there is no packet available at the AP. In the case where the buffer size is greater than 1, we can use the same method to derive the solution, and the procedures and results are omitted in the paper for simplicity.},
\begin{itemize}
\item \emph{Case 1:} There is no packet to be transmitted by the AP and only one packet request arrives. The new coming packet request will not be served until the AP obtains the available channel in the following slots.
\item \emph{Case 2:} There is no packet to be transmitted by the AP and more than one packet requests arrive. In this case only one packet request will be retained at the AP and other packet requests will be dropped. The retained request will be served when the AP obtains the available channel.
\item \emph{Case 3:} There is one packet to be transmitted by the AP. In this case, it will prevent other packet requests from entering the AP and all the arriving packet requests will be dropped.
\end{itemize}
If there is no packet waiting to be transmitted by the AP, the AP will not participate in competing for the channel.

\subsection{Node State Transition Model}
\begin{figure}[t]
\centering
\includegraphics[width=7cm]{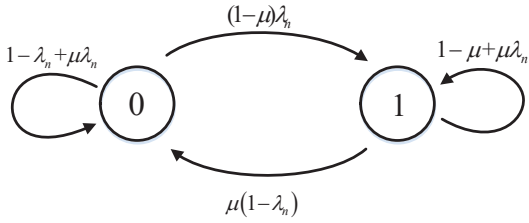}
\caption{Markov chain of the AP state.}
\label{model}
\end{figure}

Assume the request of each user is triggered according to Bernoulli distribution with parameter $\lambda$ in each slot, and one packet is demanded by each request. Based on the three cases aforementioned, the packet request at an AP with $n$ users is also a Bernoulli process with parameter $\lambda_n=1-(1-\lambda)^n$. Consider that a request is dropped out of the buffer at the end of the slot assigned to it. The PLR will be elaborated in Section \ref{pro}. Therefore, the Markov chain shown in Fig.\ref{model} is constructed, where state 0 and 1  respectively represent the cases that the AP has no and one packet to transmit. $\mu$ is the service rate of the AP, which also represents the probability that the AP successfully obtains the channel access opportunity. Then the node state balance equations are,
\begin{equation}
\left\{
\begin{aligned}
&(1-\mu)\lambda_n \pi_0 = \mu(1-\lambda_n)\pi_1 \\
&\pi_0+\pi_1 =1 \\
\end{aligned}
\right.
\end{equation}
where $\pi_0$ is the AP-empty probability indicating that there is no packet to be transmitted at the AP, and $\pi_1$ is the probability that there is one packet to be transmitted. Therefore,
\begin{equation}\label{eq9}
\mu=\frac{\lambda_n\pi_0}{\lambda_n \pi_0+(1-\lambda_n)(1-\pi_0)}.
\end{equation}
\begin{proposition}
The expectation of $\mu$ can be calculated by
\begin{equation}\label{eq2}
\mathbb{E}(\mu)=\sum\limits_{n=0}^{\infty}\frac{\frac{\lambda_1^n (K\lambda_2)^K}{(\lambda_1+K\lambda_2)^{K+n}}\frac{\Gamma(K+n)}{\Gamma(n+1)\Gamma(K)}}{1+\frac{1-\pi_0}{\pi_0}\big[\frac{1}{1-(1-\lambda)^n}-1\big]}.
\end{equation}
\end{proposition}
\begin{proof}
Based on (\ref{eq9}),  the expectation of $\mu$, with respect to the number of users associated with a typical AP, is
\begin{equation}\label{eq4}
\begin{split}
\mathbb{E}_n(\mu)&=\mathbb{E}_n\Big[\frac{\lambda_n\pi_0}{\lambda_n \pi_0+(1-\lambda_n)(1-\pi_0)}\Big] \\
&=\mathbb{E}_n\Big[\frac{1}{1+(\frac{1}{\pi_0}-1)(\frac{1}{1-(1-\lambda)^n}-1)}\Big] \\
&=\sum\limits_{n=0}^{\infty}\frac{P_n}{1+(\frac{1}{\pi_0}-1)\big[\frac{1}{1-(1-\lambda)^n}-1\big]},
\end{split}
\end{equation}
where $P_n$ is the probability that there are $n$ users associated with the typical AP. According to \cite{paper12} and with the property of PPP, the size of the Voronoi cell area is a random variable, 
and the probability mass function (PMF) of the number of users $N$ associated with a typical AP cell is
\begin{equation}\label{eq5}
\begin{split}
P_n\triangleq\mathbb{P}(N=n)=\frac{\lambda_1^n (K\lambda_2)^K}{(\lambda_1+K\lambda_2)^{K+n}}\frac{\Gamma(K+n)}{\Gamma(n+1)\Gamma(K)}.
\end{split}
\end{equation}
Based on (\ref{eq4}) and (\ref{eq5}), we have the expectation of $\mu$ in (\ref{eq4}).
\end{proof}
\section{AP-idle probability and packet loss rate}\label{pro}
In this section, we first clarify AP-idle probability. Then the PLR of the stochastic network is further elaborated.
\subsection{AP-idle Probability}
We first give another approach to obtain the expectation of $\mu$. Without loss of generality, consider that there is a typical AP located at the origin. The probability $P_{m-1}$ that there are $m-1$ other APs in the interference suppression region of the typical AP is \cite{paper10,paper1}
\begin{equation}
P_{m-1}=\frac{e^{-\lambda_2 A}(\lambda_2 A)^{m-1}}{(m-1)!}.
\end{equation}
Denote $P_{i|m}$ as the probability that there are $i$ $(1 \leq i \leq m)$ APs (including the typical AP) competing for the channel, given that there are totally $m$ APs in the interference suppression region of the typical AP. $P_{i|m}$ can be calculated as
\begin{equation}
P_{i|m}=\binom{m-1}{i-1}(1-\pi_0)^{i-1} \pi_0^{m-i}.
\end{equation}
Consider each AP has the equal probability to access the channel. Then each AP can successfully access the channel with the probability of $1/i$.
Therefore, the service rate of the AP, $\mu$, is the expectation of the probability that the AP successfully accesses the channel and is given by
\begin{align}\label{eq0}
&\mathbb{E}_m(\mu)\nonumber\\ &=\sum\limits_{m=1}^{\infty}P_{m-1}\sum\limits_{i=1}^{m}P_{i|m} \frac{1}{i}\nonumber \\
&=\sum\limits_{m=1}^{\infty}\!\sum\limits_{i=1}^{m}\frac{e^{-\lambda_2 A}(\lambda_2 A)^{m\!-\!1}}{i(m-1)!}\binom{m\!-\!1}{i\!-\!1}(1\!-\!\pi_0)^{i\!-\!1} \pi_0^{m-i} \nonumber\\
&=\sum\limits_{m=1}^{\infty}\!\sum\limits_{i=1}^{m}\frac{e^{-\lambda_2 A}(\lambda_2 A)^{m-1}\pi_0^m}{i(m\!-\!1)!(1-\pi_0)}\binom{m\!-\!1}{i\!-\!1}(1\!-\!\pi_0)^i \pi_0^{\!-i}.
\end{align}

We then have the following proposition,
\begin{proposition}\label{lemma}
The expectation of service rate $\mu$ is,
\begin{equation}\label{eq3}
\mathbb{E}(\mu)=\frac{1-e^{-\lambda_2 A}}{(1-\pi_0)\lambda_2 A}.
\end{equation}
\end{proposition}
\begin{proof}
Let $x=\frac{1-\pi_0}{\pi_0}$, and denote the second summation term of (\ref{eq0}) as $f(x)$, i.e.,
\begin{equation}
\begin{split}
f(x)&\triangleq\sum\limits_{i=1}^{m}\binom{m-1}{i-1}(1-\pi_0)^i \pi_0^{-i} \frac{1}{i} \\
&=\sum\limits_{i=1}^{m}\binom{m-1}{i-1}\frac{x^i}{i}.
\end{split}
\end{equation}
Take the derivative of $f(x)$ with respect to $x$, we have
\begin{equation}
\begin{split}
\frac{\mathrm{d}{f(x)}}{\mathrm{d}{x}}&=\sum\limits_{i=1}^{m}\binom{m-1}{i-1}x^{i-1} \\
&=\sum\limits_{i=1}^{m}\binom{m-1}{i-1}x^{i-1}\cdot 1^{m-i} \\
&\overset{(a)}{=}(x+1)^{m-1}.
\end{split}
\end{equation}
Step $(a)$ is obtained according to the binomial expansion. Taking the integral of the equation above, we get
\begin{equation}\label{eq1}
f(x)=\frac{(x+1)^m}{m}.
\end{equation}
Substituting (\ref{eq1}) in (\ref{eq0}), we obtain
\begin{equation}\label{eq6}
\begin{split}
\mathbb{E}_m(\mu) &=\sum\limits_{m=1}^{\infty}\frac{e^{-\lambda_2 A}(\lambda_2 A)^{m-1}}{(m-1)!}\frac{\pi_0^m}{1-\pi_0}\frac{(x+1)^m}{m} \\
&=\frac{e^{-\lambda_2 A}}{(1-\pi_0)\lambda_2 A}\sum\limits_{m=1}^{\infty}\frac{[\lambda_2 A \pi_0 (x+1)]^{m}}{m!} \\
&\overset{(b)}{=}\frac{e^{-\lambda_2 A}}{(1-\pi_0)\lambda_2 A}\big[e^{\lambda_2 A\pi_0(x+1)}-1\big] \\
&\overset{(c)}{=}\frac{e^{-\lambda_2 A}}{(1-\pi_0)\lambda_2 A}(e^{\lambda_2 A}-1)=\frac{1-e^{-\lambda_2 A}}{(1-\pi_0)\lambda_2 A}.
\end{split}
\end{equation}
Step $(b)$ is obtained according to the Taylor expansion. Step $(c)$ is obtained by substituting $x=\frac{1-\pi_0}{\pi_0}$ into the expression. Then the proof is finished.
\end{proof}

As (\ref{eq2}) and (\ref{eq3}) are for the expectation of $\mu$ from different perspectives, they should be equal, i.e.,
\begin{equation}\label{eq13}
\begin{split}
&\frac{1-e^{-\lambda_2 A}}{(1-\pi_0)\lambda_2 A}=\sum\limits_{n=0}^{\infty}\frac{\frac{\lambda_1^n (K\lambda_2)^K}{(\lambda_1+K\lambda_2)^{K+n}}\frac{\Gamma(K+n)}{\Gamma(n+1)\Gamma(K)}}{1+\frac{1-\pi_0}{\pi_0}\big[\frac{1}{1-(1-\lambda)^n}-1\big]}.
\end{split}
\end{equation}

With (\ref{eq13}), we obtain the value of the AP-empty probability $\pi_0$ of a randomly selected AP. Different from the AP-empty probability $\pi_0$, we denote $P_{ai}$ as the AP-idle probability, which implies the probability that an AP is not transmitting packet in a slot. It aslo indicates the percentage of APs that do not transmit packets during a slot. We shall note that the AP-idle probability is jointly decided by two factors: i) AP-empty probability, i.e., the probability that there is no request to be served by the AP; ii) the probability that there is request to be served but the AP fails to access the channel due to the  interference suppression strategy. We then have,
\begin{proposition}
The AP-idle probability $P_{ai}$ is given by
\begin{equation}\label{CC}
P_{ai}=\frac{\lambda_2 A+e^{-\lambda_2 A}-1}{\lambda_2 A}.
\end{equation}
\end{proposition}
\begin{proof}
When there is a packet to be transmitted by the AP, the probability of a successful channel access  equals to the service rate $\mu$, thereby the probability of a failed channel access is $1-\mu$. Therefore,  the probability that there is request to be served but the AP fails to access the channel is
\begin{equation}
P_{af}=\pi_1(1-\mu)=(1-\pi_0)(1-\mu).
\end{equation}
Based on the two factors of AP-idle probability, we then have
\begin{equation}
\begin{split}
P_{ai}&=\pi_0+P_{af} \\
&=1-\mu+\pi_0\mu \\
&=\frac{\lambda_2 A+e^{-\lambda_2 A}-1}{\lambda_2 A}.
\end{split}
\end{equation}
We obtain (\ref{CC}) and the proof is finished.
\end{proof}


\subsection{Packet Loss Rate}
As elaborated above, a packet needs to be completely transmitted during the slot assigned to it, otherwise it will be dropped. The successful transmission probability can be defined as,
\begin{equation}
\begin{split}
p&\triangleq \mathbb{P}[\tau B \log_2(1+\sf{SINR}) \geq \it{T}] \\
&=\mathbb{P}(\sf{SINR} \geq 2^{\frac{\it{T}}{\it{\tau} \it{B}}}-1),
\end{split}
\end{equation}
where $\it{B}$ is the bandwidth. Denote $\bar{T} \triangleq 2^{\frac{\it{T}}{\it{\tau} \it{B}}}-1$ hereafter.
Given the AP-idle probability, the interfering APs can be modeled as thinning PPP outside the interference suppression region with intensity $(1-P_{ai})\lambda_2$. Based on \cite{paper12}, the probability that a randomly chosen user can achieve a target SINR is,
\begin{equation}
\begin{split}
p=&\mathbb{P}(\sf{SINR} \geq \it{\bar{T}}) \\
=&\int_{0}^{\infty}2\pi\lambda_2r \exp\Big\{-\lambda_2\pi r^2-(1-P_{ai})\lambda_2 \pi r^2 \sqrt{\bar{T}} \\
&\arctan\Big[\frac{r^2 \sqrt{\bar{T}}}{(r_s-r)^2}\Big]-\frac{\bar{T}r^{\eta}\sigma^2}{P}\Big\}\mathrm{d}r,
\end{split}
\end{equation}
where $\sigma^2$ is the variance of the addictive white Gaussian noise (AWGN), and $\eta$ is the path-loss exponent ($\eta > 2$).

Packets will be lost in the following scenarios: i) more than one packet requests arrive at the empty AP; ii) one or more packet requests arrive at the AP with a packet to transmit; iii) the received  SINR falling below the threshold when a packet is transmitted. On condition that there are $j$ packet requests arriving at the AP in one slot, when there is no packet being transmitted by the AP, the PLR of a random packet is $\frac{j-1}{j}$. When there is one packet being transmitted by the AP, if the typical AP successfully accesses the channel, the PLR is $1-p+p\frac{j}{j+1}$; otherwise, the PLR is $\frac{j}{j+1}$. Therefore, we obtain the average PLR of the whole network, which is calculated by (\ref{eq14}) in the next page.
\begin{figure*}[t]
\begin{align}\label{eq14}
p_{l}=&\pi_0\Bigg[\sum\limits_{i=1}^{\infty} P_i\sum\limits_{j=1}^{i}\binom{i}{j}\lambda^j(1-\lambda)^{i-j}\frac{j-1}{j}\Bigg]
+\pi_1\Bigg\{\sum\limits_{i=1}^{\infty} P_i\sum\limits_{j=0}^{i}\binom{i}{j}\lambda^j(1-\lambda)^{i-j}
\bigg[\mu\Big(1-p+p\frac{j}{j+1}\Big)+(1-\mu)\frac{j}{j+1}\bigg]\Bigg\} \nonumber\\
\overset{(d)}{=}&\pi_0\Bigg[\sum\limits_{i=1}^{\infty}\frac{\lambda_1^i (K\lambda_2)^K}{(\lambda_1+K\lambda_2)^{K+i}}\frac{\Gamma(K+i)}{\Gamma(i+1)\Gamma(K)}
\sum\limits_{j=1}^{i}\binom{i}{j}\lambda^j(1-\lambda)^{i-j}\frac{j-1}{j}\Bigg] \nonumber\\
+&(1-\pi_0)\Bigg\{\sum\limits_{i=1}^{\infty} \frac{\lambda_1^i (K\lambda_2)^K}{(\lambda_1+K\lambda_2)^{K+i}}\frac{\Gamma(K+i)}{\Gamma(i+1)\Gamma(K)}
\sum\limits_{j=0}^{i}\binom{i}{j}\lambda^j(1-\lambda)^{i-j}\bigg[\mu\Big(1-\frac{p}{j+1}\Big)+\frac{(1-\mu)j}{j+1}\bigg]\Bigg\}.
\end{align}
\hrulefill
\end{figure*}
Here, $P_i$ denotes the probability that there are $i$ packet requests arriving at the AP in a slot and step $(d)$ is obtained according to (\ref{eq5}).

It can be observed that the PLR is jointly decided by the network property (e.g., the user density $\lambda_1$, AP density $\lambda_2$, the transmission power $P$ and the path-loss exponent $\eta$) and  the service property (e.g., the packet request rate $\lambda$, packet size $T$ and the slot duration $\tau$).
\section{numerical results}\label{num}
\begin{table}[t]
\centering
\caption{Typical Parameters For the Simulation}
\begin{tabular}{|c|c|c|}
\hline
Symbols & Notations & Values \\ \hline
$\tau$ &Slot duration&  0.5 s\\ \hline
$B$ &Bandwidth&  10 MHz\\ \hline
$T$ &Packet size&  10 Mbits\\ \hline
$R$ &Radius of interference suppression&  250 m\\ \hline
$\eta$ & Path-loss exponent &  4\\ \hline
$\sigma^2$ & The power of AWGN & 0\\ \hline
$P$ & Transmit power & 33 dBm\\ \hline
$\lambda_1$ & User density & $1000$ nodes/km$^2$\\ \hline
$\lambda_2$ & AP density & $100$ nodes/km$^2$\\ \hline
$\lambda$ & Bernoulli distribution parameter & $0.03$\\ \hline
\end{tabular}
\end{table}

In this section, we verify the analytical results of the AP-empty probability and the PLR with Monte-Carlo simulations. The simulation results are obtained in a square area of 2km $\times$ 2km. Theoretical results and simulation results are illustrated for both the general case and the case with fully loaded APs. Typical parameters are listed in Table I, and they do not change  unless additional notations are clarified.


\begin{figure}[t]
\centering
\includegraphics[width=8.2cm]{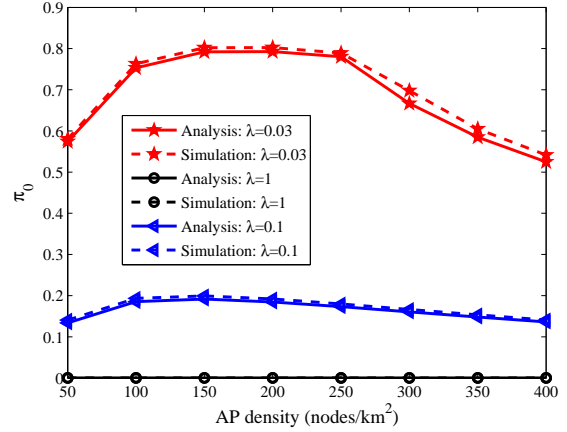}
\caption{AP-empty probability versus AP density.}
\label{fig2}
\end{figure}
First, we conduct the simulation to verify the influence of the AP density on the AP-empty probability. The general case and the case that the APs are full-loaded ($\lambda=1$) are presented in Fig. \ref{fig2}. The simulation results are well consistent with the analytical results. As illustrated in Fig. \ref{fig2}, the AP-empty probability first increases with the increase of the AP density. When the AP density exceeds a critical value, the curve starts to descend and the AP-empty probability decreases with the increase of the AP density. This is due to the fact that when the density of APs is low, the increasing number of APs reduces the average number of users associated with each AP. Therefore, request arrivals at each AP decrease, which enhances the AP-empty probability. However, when the AP density reaches the critical value, the spectrum becomes the bottleneck of the system performance. Due to the limited spectrum and the interference suppression strategy, not every AP can access the channel to transmit its packets in each slot. The increasing number of APs makes it more difficult to access the channel. Therefore, more APs stay idle during one slot. Furthermore, it can be observed that the AP-empty probability decreases with the increase of users' request rate. This is due to the fact that larger $\lambda$ (larger request rate) implies more packet request arrivals at the AP in each slot, increasing the probability of the AP to be occupied.

\begin{figure}[t]
\centering
\includegraphics[width=8.2cm]{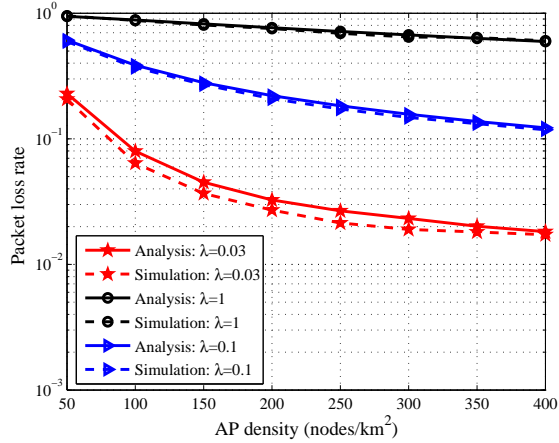}
\caption{Packet loss rate versus the AP density.}
\label{fig3}
\end{figure}

Next, the PLR with respect to the AP density is illustrated in Fig. \ref{fig3}. Compared with the case where the APs are fully loaded ($\lambda=1$), the PLR of the general case ($\lambda=0.03, 0.1$) decreases because of less packet request arrivals during each slot. Moreover, the PLR decreases with the increase of the AP density. This is due to that the increasing number of APs reduces the number of users associated with each AP. Consequently, packets can be transmitted via more APs simultaneously and less packets will be dropped in each slot.

\begin{figure}[t]
\centering
\includegraphics[width=8.2cm]{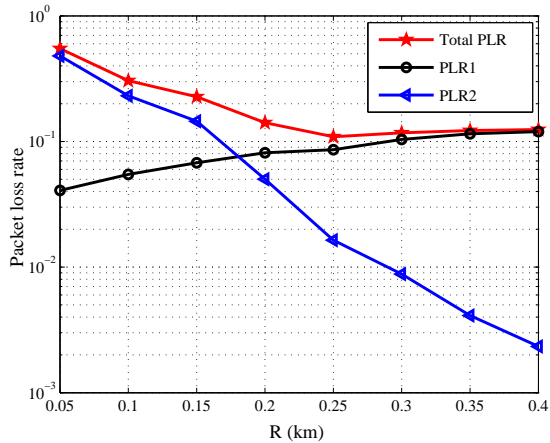}
\caption{Packet loss rate versus the radius $R$.}
\label{fig4}
\end{figure}

Fig. \ref{fig4} reveals the relationship between the PLR and the radius ($R$) of the interference suppression region. PLR1 is the packet loss rate due to the failure of the AP to obtain the channel access opportunity, and PLR2 is that due to the poor quality of the signal when the received SINR falls below the threshold. It can be observed that PLR1 increases with the increase of $R$. This is because when $R$ becomes larger, more APs keep silent, so that more packet arrivals are dropped. However, on the contrary, PLR2 decreases with the increase of $R$. This is due to the fact that larger $R$ enlarges the distance between the simultaneously active APs, reducing the interference to users. Moreover, in our model, the total PLR first decreases and then increases with the increase of $R$. When the radius of the interference suppression region is $250$m, the PLR reaches the minimum value.

\section{conclusion}\label{conc}
Different from most existing works which assume full-load cells, this paper studies the performance of the stochastic network with dynamic traffic. The interference suppression strategy is introduced to reduce the serious interference from the adjacent APs. The influence of the interference suppression region is clarified.  We propose a novel analysis on the intensity of simultaneously active APs based on the AP-idle probability.  We consider the AP-idle probability jointly decided by two factors (having no packets to transmit and failing to obtain the channel access opportunity). The PLR due to the interference (the SINR is lower than the threshold), the limited spectrum and the limited buffer size  is further elaborated.  Simulation results validate the analytical results well.

\bibliographystyle{IEEEtran}
\bibliography{paper}
\end{document}